\documentclass{article}[11pt]
\usepackage{fullpage}
\usepackage{amsmath}
\usepackage{amssymb}
\usepackage{newtxtext, newtxmath}

\usepackage{amsthm}
\usepackage[english]{babel}
\usepackage{caption}
\usepackage{xcolor}
\usepackage{algorithm}
\usepackage{algpseudocode}
\usepackage{tikz}
\usepackage{listings}
\usepackage{hyperref}
\usetikzlibrary{matrix,arrows.meta,positioning,calc,decorations.markings,decorations.text}
\newtheorem{theorem}{Theorem}
\newtheorem{lemma}[theorem]{Lemma}
\newtheorem{claim}[theorem]{Claim}
\newtheorem{prop}[theorem]{Proposition}
\newtheorem{conjecture}{Conjecture}
\newenvironment{claimproof}[1]{%
  \begin{proof}[Proof of Claim.]#1%
}{%
  \end{proof}%
}

\tikzset{>={Latex[length=3pt,width=2.6pt]}}

\newcommand{\R}{\mathbb{R}}
\newcommand{\eps}{\varepsilon}
\newcommand{\floor}[1]{\left\lfloor#1\right\rfloor}
\newcommand{\ceil}[1]{\left\lceil#1\right\rceil}

\DeclareMathOperator{\OPT}{OPT}

\lstdefinelanguage{Julia}%
  {morekeywords={abstract,break,case,catch,const,continue,do,else,elseif,end,
      export,false,for,function,immutable,import,importall,if,in,
      macro,module,quote,return,switch,true,try,type,typealias,
      using,while},%
   sensitive=true,%
   morecomment=[l]\#,%
   morestring=[b]",%
}[keywords,comments,strings]%

\title{Weighted Chairman Assignment and Flow-Time Scheduling}

\author{%
  Siyue Liu\thanks{Tepper School of Business, Carnegie Mellon University.
  Email: \href{mailto:siyueliu@andrew.cmu.edu}{\texttt{siyueliu@andrew.cmu.edu}}.}
  \and
  Victor Reis\thanks{Microsoft Research. Email: \href{mailto:victorol@microsoft.com}{\texttt{victorol@microsoft.com}}.}
}
\date{}

\begin{document}

\maketitle

\begin{abstract}
Given positive integers $m, n$, a fractional assignment $x \in [0,1]^{m \times n}$ and weights $d \in \R^n_{>0}$, we show that there exists an assignment $y \in \{0,1\}^{m \times n}$ so that for every $i\in[m]$ and $t\in [n]$, \[ \Big|\sum_{j \in [t]} d_j (x_{ij} - y_{ij}) \Big| < \max_{j \in [n]} d_j. \]
This generalizes a result of Tijdeman (1973) on the unweighted version, known as the chairman assignment problem. This also confirms a special case of the single-source unsplittable flow conjecture with arc-wise lower and upper bounds due to Morell and Skutella (IPCO 2020). As an application, we consider a scheduling problem where jobs have release times and machines have closing times, and a job can only be scheduled on a machine if it is released before the machine closes. We give a $3$-approximation algorithm for maximum flow-time minimization.

\end{abstract}

\section{Introduction}

Let $m,n$ be positive integers. We say that a matrix $x \in [0,1]^{m \times n}$ is a \emph{fractional assignment} if $\sum_{i \in [m]} x_{ij} = 1$ for every column $j \in [n]$, and that $y \in \{0,1\}^{m \times n}$ is an \emph{(integral) assignment} if $\sum_{i \in [m]} y_{ij} = 1$ for every $j \in [n]$. Consider the following problem introduced by Niederreiter in 1972 \cite{Niederreiter1972UniformlyDistributed, Niederreiter1972DistributionFiniteSets}: find the minimum value of $\Delta(m)$ such that for every positive integer $n$ and fractional assignment $x \in [0,1]^{m \times n}$, there is an assignment $y \in \{0,1\}^{m \times n}$ so that for every row $i\in [m]$ and column $t \in [n]$, \[ \Big|\sum_{j \in [t]} (x_{ij} - y_{ij}) \Big| \le \Delta(m).\] 
This is known as the \emph{chairman assignment problem} due to Tijdeman \cite{Tijdeman1980Chairman}: suppose $m$ states form a union and each state $i$ receives a value $x_{ij}$ from being a member at year $j$ so that $\sum_{i\in [m]} x_{ij}=1$. Every year a union chairman has to be selected in such a way that at any year $t$ the accumulated number of chairmen $\sum_{j\in [t]} y_{ij}$ from each state $i$ is proportional to its accumulated value $\sum_{j\in [t]} x_{ij}$. Niederreiter showed a bound of $\Delta(m) \le m-1$~\cite{Niederreiter1972UniformlyDistributed} which was subsequently improved by Meijer and Niederreiter to $\Delta(m) \le O(\log m)$~\cite{MeijerNiederreiter1972FiniteSets} and by Tijdeman to $\Delta(m) \le 1$~\cite{Tijdeman1973Distribution}, who also showed that, for $m > 1$ and any $\delta > 0$, $\Delta(m) \ge 1 - \frac{1}{2m-2} - \delta$ via a family of examples with $n = \Omega(1/\delta)$. The approaches of \cite{MeijerNiederreiter1972FiniteSets} and \cite{Tijdeman1973Distribution} are both based on an application of Hall's theorem to a carefully constructed bipartite graph. Finally, Meijer matched this lower bound by showing $\Delta(m) \le 1-\frac{1}{2m-2}$~\cite{Meijer1973FiniteSets}. Since then, there have been other proofs of this result through different approaches~\cite{Tijdeman1980Chairman,AngelHolroydMartinPropp2009DiscreteLDS,HolroydPropp2010RotorWalks, BertheCartonChevallierSteinerYassawi2024}; see also the survey of Tijdeman~\cite{Tijdeman1982ProgressDiscrepancy}.


We study the \emph{weighted chairman assignment problem} where each $j \in [n]$ has an associated weight $d_j > 0$ and we want to bound $|\sum_{j \in [t]} d_j (x_{ij} - y_{ij}) |$ for every $i\in[m]$ and $t\in[n]$. All of the previous approaches in the unweighted setting rely on the fact that the increments to the integral assignment are integers and do not directly generalize to the weighted setting. Our main theorem generalizes an algorithm of Tijdeman \cite{Tijdeman1980Chairman} (see also Angel, Holroyd, Martin and Propp~\cite{AngelHolroydMartinPropp2009DiscreteLDS}) with a new analysis to show the following:

\begin{theorem}\label{thm:WCAP}
Given positive integers $m, n$ with $m > 1$, a fractional assignment $x \in [0,1]^{m \times n}$ and weights $d \in \R^n_{>0}$, there is a linear-time algorithm that computes an assignment $y \in \{0,1\}^{m \times n}$ so that for every $i\in[m]$ and $t\in [n]$,
\[ \Big|\sum_{j \in [t]} d_j (x_{ij} - y_{ij}) \Big| \le \Big(1-\frac{1}{2m-2}\Big) \cdot \max_{j \in [n]} d_j. \]
\end{theorem}

The weighted chairman assignment problem gains renewed interest due to a conjecture by Morell and Skutella \cite{MorellSkutella2020IPCO}: given a flow $x$ from a single source to multiple sinks with varying demands, there is an \emph{unsplittable flow} $y$ where each sink should be served by a single path, such that the discrepancy between the two flows on each arc is at most the maximum demand. The weighted chairman assignment problem with weights $d_1,\ldots, d_n$ can be modeled as a single-source unsplittable flow instance with demands $d_1,\ldots, d_n$ (see  Section \ref{sec:related_work}), and therefore Theorem \ref{thm:WCAP} with discrepancy bound $\max_{j\in [n]} d_j$ would follow if the conjecture were true. When the demands are uniform, i.e., $d_1=\ldots=d_n=1$, the conjecture is true because by adding capacity constraints $\floor{x}\leq y\leq \ceil{x}$ to the flow polytope we can find an \emph{integral} flow $y$ which is automatically unsplittable and satisfies $|x-y|_1\leq 1$. Therefore, the discrepancy bound $1$ for the (unweighted) chairman assignment follows as a consequence. The weighted setting corresponds to nonuniform demands, where so far all linear programming techniques fall apart.

As an application, we consider the following scheduling problem. Suppose there are $m$ machines $M$ and $n$ jobs $J$. Each machine $i \in M$ has a closing time $b_i \ge 0$. Each job $j \in J$ has a release time $r_j \ge 0$ and a processing time $d_j \ge 0$. Job $j$ can be scheduled on machine $i$ if and only if $r_j \le b_i$. The \emph{flow-time} of a job $j$ is defined as the time elapsed from its release to completion. The goal is to schedule jobs to machines to minimize the maximum flow-time of any job. This is a special case of the \emph{restricted assignment} variant of maximum flow-time minimization, where each job $j$ has a subset of machines $M_j\subseteq M$ that it can be scheduled on. In the most general form of the problem, each job $j$ has a potentially different processing time $d_{ij}$ on machines $i\in [m]$. 
Bansal, Rohwedder and Svensson~\cite{BansalRohwedderSvensson2022STOC} proved that a natural linear programming relaxation has integrality gap $O(\sqrt{\log n})$, and conjectured that an $O(1)$-approximation should be possible. We confirm this conjecture for the setting of machine closing times:

\begin{theorem}\label{thm:closing_time}
There is a $(3-\frac{1}{m-1})$-approximation algorithm for the maximum flow-time minimization problem in the setting of machine closing times.
\end{theorem}

Previously, a $(3-\frac{2}{m})$-approximation was known via a greedy algorithm (first-in-first-out, or FIFO) for the setting of identical machines where $b_i = \infty$ for all $i \in M$~\cite{BenderChakrabartiMuthukrishnan1998SODA, Mastrolilli2004IJFCS}, and this bound is known to be tight for FIFO~\cite{Mastrolilli2003Computing}. We also show that FIFO has approximation ratio at least $\Omega(\log m)$ for the setting of closing times.

Furthermore, we give several lower bound constructions. We provide a construction simpler than Tijdeman's~\cite{Tijdeman1973Distribution} to establish a matching lower bound for the chairman assignment problem:

\begin{prop} \label{prop:caplb}
For any positive integer $m > 1$ there exists a fractional assignment $x \in [0,1]^{m \times (m-1)}$ so that for any assignment $y \in \{0,1\}^{m \times (m-1)}$, there exists some $i\in [m]$ and $t \in [m-1]$ for which \[ \Big|\sum_{j \in [t]} (x_{ij} - y_{ij}) \Big| \ge 1-\frac{1}{2m-2}. \]
\end{prop}

The assignment given in~\cite{Tijdeman1973Distribution} that achieves $\Delta(m)\le 1$ yields a strict inequality and has the property that $y_{ij} = 0$ whenever $x_{ij} = 0$. We give a construction showing that Tijdeman's bound is tight under this additional constraint:

\begin{prop} \label{prop:carlb}
For any $\delta > 0$ there exists a positive integer $n$ and a fractional assignment $x \in [0,1]^{3 \times n}$ so that for any assignment $y \in \{0,1\}^{3 \times n}$ satisfying $y_{ij} = 0$ whenever $x_{ij} = 0$, there exists some $i\in [3]$ and $t \in [n]$ for which \[ \Big|\sum_{j \in [t]} (x_{ij} - y_{ij}) \Big| \ge 1-\delta. \]
\end{prop}

Finally, we disprove a conjecture in~\cite{AngelHolroydMartinPropp2009DiscreteLDS} that one may obtain a bound of 1 for all intervals, not just prefixes:

\begin{prop} \label{prop:intlb}
There exist positive integers $m,n$ and a fractional assignment $x \in [0,1]^{m \times n}$, so that for any assignment $y \in \{0,1\}^{m \times n}$, there exist some $i\in [m]$ and $1 \le s \le t \le n$ for which \[ \Big|\sum_{j \in [s,t]} (x_{ij} - y_{ij}) \Big| > 1. \]
\end{prop}

In particular, this implies that we cannot hope to directly apply our rounding approach to obtain a 2-approximation for maximum flow-time minimization.

\paragraph{Organization} In Section \ref{sec:related_work}, we discuss related work. In Section \ref{sec:proof}, we prove our main Theorem \ref{thm:WCAP}. In Section \ref{sec:schedule}, we discuss connections to maximum flow-time scheduling and prove Theorem \ref{thm:closing_time}. In Section \ref{sec:lower_bound}, we give lower bound constructions to prove Propositions \ref{prop:caplb}, \ref{prop:carlb} and \ref{prop:intlb}. In Section \ref{sec:open}, we formulate some open questions. 

\section{Related work}\label{sec:related_work}

Theorem~\ref{thm:WCAP} resolves a special case of a conjecture of Morell and Skutella~\cite{MorellSkutella2020IPCO} which can be formally stated as follows.
Let $D=(V,A)$ be a directed acyclic graph with source $s \in V$, sink terminals $t_1, \dots, t_n \in V$ and associated demands $d \in \R^n_{> 0}$. We say a flow $x \in \R^{A}_{\ge 0}$ satisfies the demands if $x(\delta^-(t_j))=d_j$ for every $j\in [n]$, and $x(\delta^-(v))=x(\delta^+(v))$ for every $v\in V\setminus \{s,t_1,...,t_n\}$. We say a flow $y \in \R^{A}_{\ge 0}$ is \emph{unsplittable} if for each $j \in [n]$ there is a single path $P_j$ from $s$ to $t_j$ that carries $d_j$ units of flow, so that $y_a=\sum_{j: a\in P_j} d_j$ for every $a\in A$.
Morell and Skutella conjectured the following:
\begin{conjecture}[\cite{MorellSkutella2020IPCO}]\label{conj:ssuf}
   For every flow $x \in \R^{A}_{\ge 0}$ satisfying the demands, there is an unsplittable flow $y \in \R^{A}_{\ge 0}$ such that $|x_a-y_a|\le \max_{j\in [n]} d_j$ for all $a \in A$.
\end{conjecture}
They also showed the one-sided bound $x_a - y_a \le \max_{j \in [n]} d_j$ for all $a \in A$ may be satisfied, thereby complementing the previously known one-sided bound $y_a - x_a \le \max_{j \in [n]} d_j$  for all $a \in A$ of Dinitz, Garg and Goemans~\cite{DinitzGargGoemans1999Combinatorica}. So far, Conjecture~\ref{conj:ssuf} has been established when the demands have the property that one divides another, i.e., $d_1\mid d_2\mid \ldots \mid d_n$ \cite{MorellSkutella2020IPCO} and for special digraphs such as acyclic planar digraphs~\cite{TraubVargasKochZenklusen2024SODA}; see also~\cite{AlmoghrabiSkutellaWarode2025IPCO} for stronger guarantees for series-parallel digraphs.  

\begin{figure}[H]
\centering
\begin{tikzpicture}[
  x=0.7cm, y=1.2cm,              
  edge/.style={-Latex, line width=.9pt},
  toT/.style={shorten >=6pt},    
  term/.style={rectangle, minimum size=7pt, inner sep=0pt, draw=black, fill=black},
  src/.style={circle, minimum size=7pt, inner sep=0pt, draw=black, fill=black},
  lab/.style={font=\large}
]

\coordinate (s)  at (0,1);

\coordinate (a1) at (4,2);   
\coordinate (b1) at (4,-1.1);   
\coordinate (d1) at (4,-0.3);
\coordinate (c1) at (4, 0.5);
\coordinate (c2) at (8, 0.5);
\coordinate (c3) at (12, 0.5);

\coordinate (a2) at (8,2);   
\coordinate (b2) at (8,-1.1);
\coordinate (d2) at (8,-0.3);

\coordinate (a3) at (12,2);   
\coordinate (b3) at (12,-1.1);
\coordinate (d3) at (12,-0.3);

\coordinate (t3) at (7,1);   
\coordinate (t2) at (11,1);
\coordinate (t1) at (15,1);

\node[src] (S) at (s) {};
\node[term] (T3) at (t3) {};
\node[term] (T2) at (t2) {};
\node[term] (T1) at (t1) {};

\node[lab,left=5pt of S] {$s$};
\node[lab,font=\itshape,above right=2pt and 0pt of S] {$\displaystyle\sum_{j=1}^{3} x_{ij} d_j$};

\node[lab,right=4pt of T3] {$t_3$};
\node[lab,right=4pt of T2] {$t_2$};
\node[lab,right=4pt of T1] {$t_1$};
\node[lab,font=\itshape, above right=-1pt of a3] {$i_1$};

\node[lab,font=\itshape, above right=-1pt of a2] {$i_2$};

\node[lab,font=\itshape, above left=-1pt of a1] {$i_3$};

\draw[edge] (S) -- (a1);
\draw[edge] (S) -- (b1);
\draw[edge] (S) -- (c1);
\draw[edge] (S) -- (d1);

\draw[edge] (a1) -- (a2)
  node[midway, above] {$x_{i1} d_1 + x_{i2} d_2$};
\draw[edge] (a2) -- (a3)
  node[midway, above] {$x_{i1} d_1$};

\draw[edge] (b1) -- (b2);
\draw[edge] (b2) -- (b3);
\draw[edge] (c1) -- (c2);
\draw[edge] (c2) -- (c3);
\draw[edge] (d2) -- (d3);

\draw[edge] (d1) -- (d2);

\draw[edge,toT] (a1) -- (t3) node[pos=.52, above right=-1pt] {$x_{i3} d_3$};
\draw[edge,toT] (b1) -- (7.15,1);
\draw[edge,toT] (c1) -- (7,1.03);
\draw[edge,toT] (d1) -- (t3);

\draw[edge,toT] (a2) -- (t2) node[pos=.54, above right=-1pt] {$x_{i2} d_2$};
\draw[edge,toT] (b2) -- (11.15,1);
\draw[edge,toT] (c2) -- (11,1.03);
\draw[edge,toT] (d2) -- (t2);

\draw[edge,toT] (a3) -- (t1) node[pos=.54, above right=-1pt] {$x_{i1} d_1$};
\draw[edge,toT] (b3) -- (15.15,1);
\draw[edge,toT] (c3) -- (15,1.03);
\draw[edge,toT] (d3) -- (t1);

\end{tikzpicture}
\caption{Reduction for $m = 4$, $n = 3$. We create $n$ copies $i_1, \dots, i_n$ of each $i \in [m]$ to capture all prefixes. Terminal $t_j$ has demand $\sum_{i\in[m]} x_{ij}d_j=d_j$, and is connected to $i_j$ for each $i\in [m]$. The flow value on arc $(i_{t+1},i_t)$ is $\sum_{j\in[t]} x_{ij}d_j$.}
\label{fig:ssuf}
\end{figure}
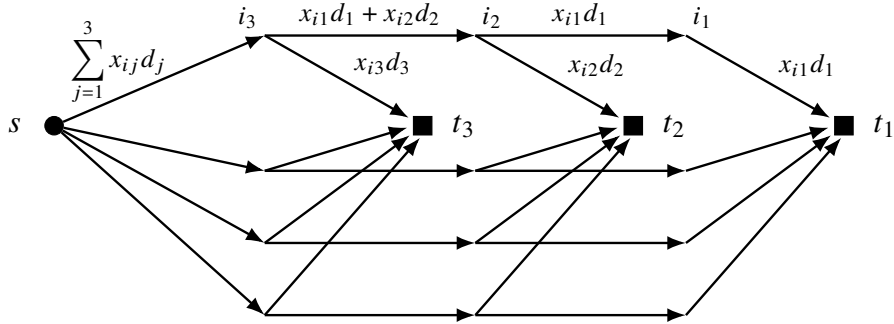

To see the connection with Theorem~\ref{thm:WCAP}, note that given positive integers $m, n$ we may construct a directed acyclic graph consisting of $m$ paths of length $n$ starting from $s$, with edges pointing at $n$ terminals along the way. Then any fractional assignment $x \in [0,1]^{m \times n}$ and weights $d \in \R^{n}_{>0}$ induce a flow so that each terminal $t_j$ has demand $d_j$ and each prefix is captured by one of the arcs in a path, as illustrated in Figure \ref{fig:ssuf}. An unsplittable flow routes each demand $d_j$ to $t_j$ through a path $i\in[m]$, which corresponds to assigning $j$ to $i$. Assuming Conjecture \ref{conj:ssuf}, such an assignment $y$ satisfies $|\sum_{j\in[t]} d_j(x_{ij}-y_{ij})|\le \max_{j\in[n]} d_j$ for all $i \in [m], t \in [n]$. 

A variant of the chairman assignment problem is the \emph{carpooling problem}~\cite{FaginWilliams1983IBM, AjtaiAspnesNaorRabaniSchulmanWaarts1995SODA}, where $j\in[n]$ can be assigned to $i\in [m]$ if and only if $x_{ij}\neq 0$. The \emph{weighted carpooling problem}, which remains open, can also be viewed as a special case of Conjecture \ref{conj:ssuf} using a similar construction that connects terminal $t_j$ to paths $i$ for which $x_{ij}\neq 0$:
\begin{conjecture}\label{conj:weighted_carpool}
   Let $m, n$ be positive integers and $d \in \R^n_{>0}$. For every fractional assignment $x \in [0,1]^{m \times n}$, there is an assignment $y \in \{0,1\}^{m \times n}$ so that $y_{ij} = 0$ whenever $x_{ij} = 0$, and for every $i\in [m]$ and $t\in [n]$,
    \[
    \Big|\sum_{j \in [t]} d_j (x_{ij}-y_{ij}) \Big| \le  \max_{j \in [n]} d_j.
    \]
\end{conjecture}

Morell and Skutella~\cite{MorellSkutella2020IPCO} outlined the connection to maximum flow-time minimization and noted that Conjecture~\ref{conj:weighted_carpool} would imply a 3-approximation for the restricted assignment setting. The best known approximation in polynomial time is $O(\log n)$ due to Bansal and Kulkarni via iterated rounding~\cite{BK15}. Bansal, Rohwedder and Svensson~\cite{BansalRohwedderSvensson2022STOC} prove a bound of $O(\sqrt{\log n})\cdot  \max_{j \in [n]} d_j$ for Conjecture \ref{conj:weighted_carpool}, leveraging a non-constructive argument from convex geometry by Banaszczyk \cite{banaszczyk2012series}, thereby showing that the natural linear programming relaxation for maximum flow-time minimization has integrality gap $O(\sqrt{\log n})$. For identical machines, a simple 3-approximation is known~\cite{BenderChakrabartiMuthukrishnan1998SODA, Mastrolilli2004IJFCS}: sort the jobs by release times and assign them one by one to the machine with the least remaining processing time of jobs that have been assigned to it. The special case where release times are all zero corresponds to makespan minimization; the work of Lenstra, Shmoys and Tardos gives a 2-approximation algorithm via rounding~\cite{lenstra1990approximation}. Their approach also solves the Morell-Skutella conjecture for digraphs of diameter 2.

To the best of our knowledge, the weighted chairman assignment problem has not been studied before. The more general weighted carpooling problem fits within the broader framework of \emph{prefix discrepancy}~\cite{BansalRohwedderSvensson2022STOC}, and it is possible to derive a bound of $2m \cdot \max_{j \in [n]} d_j$ for Conjecture \ref{conj:weighted_carpool} via a linear algebraic argument~\cite{BaranyGrinberg1981,Barany2010}. Together with Banaszczyk's result \cite{banaszczyk2012series}, the current best bound for Conjecture \ref{conj:weighted_carpool} is $\min\{2m, O(\sqrt{\log n})\} \cdot \max_{j \in [n]} d_j$. The special case of Conjecture~\ref{conj:weighted_carpool} where each column $\{x_{ij}\}_{i\in [m]}$ has at most two nonzero entries is known as the \emph{2-sparse prefix Beck-Fiala problem} and is equivalent to the general case up to a constant~\cite{BansalRohwedderSvensson2022STOC}. 
Proposition~\ref{prop:carlb} shows that unlike Theorem~\ref{thm:WCAP}, we cannot hope for a $(1-\delta) \cdot \max_{j \in [n]} d_j$ bound for Conjecture~\ref{conj:weighted_carpool} for any $\delta > 0$ even for $m=3$.

For the online version of the chairman assignment problem, where columns of $x$ arrive one at a time and the assignment must be decided irrevocably, it is known that the best possible bound is $\sum_{j=2}^m \tfrac{1}{j} = \Theta(\log m)$ for adaptive adversaries, attainable with a simple greedy algorithm~\cite{CoppersmithNowickiPaleologoTresserWu2011}, and the analysis readily generalizes for the weighted setting. On the other hand, the best possible bound independent of $n$ for the online version of the carpooling problem is $\frac{m-1}{2}$~\cite{CoppersmithNowickiPaleologoTresserWu2011}. For the online 2-sparse prefix Beck-Fiala problem, there is also an $O(\sqrt{\log n})$ bound~\cite{KRR24} and a $\Omega(\sqrt[3]{\log n})$ lower bound~\cite{AjtaiAspnesNaorRabaniSchulmanWaarts1995SODA} for oblivious adversaries. A recent line of work studies the online carpooling problem with \emph{recourse}~\cite{GuptaGurunathanKrishnaswamyKumarSingla2022, EfronPatelStein2025}, and it remains open to obtain a constant bound for the online chairman assignment problem with polylogarithmic recourse. 

\section{Proof of Theorem \ref{thm:WCAP}}\label{sec:proof}

Given $m,n \in \mathbb{N}$ with $m  > 1$, denote $\eps := \frac{1}{2m-2}$. Let $x \in [0,1]^{m \times n}$ be a fractional assignment. We normalize $d\in \R_{>0}^n$ so that $\max_{j \in [n]} d_j=1$, thereby assuming $d\in(0,1]^n$. For $j\in [n]$, denote by $s_j\in[m]$ the element to which $j$ will be assigned, so that $y_{ij}=1$ if $s_j=i$ and $0$ otherwise. For each $i \in [m]$ and $t \in [n]$, denote
\[
P_t(i):=\sum_{j=1}^t d_j x_{ij}, \quad N_t(i):=\sum_{j=1}^t d_j y_{ij},\quad \Delta_t(i) := P_t(i) - N_t(i).
\]
At time $t$, the \emph{deadline} of $i$ is defined as
\begin{equation}\label{eq:ddl}
    D_t (i) := \min\{T \ge t : P_T(i)\ge N_{t-1}(i)+1-\eps\},
\end{equation}

which is set to $\infty$ (or $n+1$) if such $T$ does not exist.

The set of \emph{candidates} at time $t$ is defined as 
\begin{equation}\label{eq:candidate}
    C(t):=\big\{i\in[m]: P_t(i)\ge N_{t-1}(i)+\min\{d_t/m,\eps\}\big\}.
\end{equation}
At time $t$, the algorithm chooses a candidate with the earliest deadline. The detailed description is in Algorithm \ref{alg:deadline}. We remark that if we only wanted a bound of $\max_{j\in [n]} d_j$ in Theorem \ref{thm:WCAP}, simpler definitions of deadline and candidates where setting $\eps=0$ in both definitions would have been sufficient.

\begin{algorithm}[H]
\caption{\textsc{Earliest Deadline Algorithm}}\label{alg:deadline}
\begin{algorithmic}
\Require Fractional assignment $x \in [0,1]^{m \times n}$, $d \in (0,1]^n$
\Ensure Assignment $y \in \{0,1\}^{m \times n}$ such that 
        $|\Delta_t(i)| \le 1-\eps$ for all $i \in [m],\ t \in [n]$

\For{$t \in [n]$}
    \State $C(t) \gets \{ i \in [m] : P_t(i) \ge N_{t-1}(i) + \min\{ d_t/m, \eps \} \}$ \Comment{candidates}
    \For{$i \in [m]$}
        \State $D_t(i) \gets \min \{ T \ge t : P_T(i) \ge N_{t-1}(i) + 1 - \eps \}$ \Comment{deadlines}
    \EndFor
    \State Choose $s_t \in \arg\min_{i \in C(t)} D_t(i)$
    \State $y_{s_t, t} \gets 1$ and $y_{i, t} \gets 0$ for $i \in [m] \setminus \{s_t\}$
\EndFor

\State Output $y$
\end{algorithmic}
\end{algorithm}

The feasibility of the algorithm follows from the lemma below.
\begin{lemma}
   For any $t \in [n]$, there always exists a candidate, i.e., $C(t)\neq \emptyset$.
\end{lemma}
\begin{proof}
    Since 
    \[
    \begin{aligned}
        \sum_{i=1}^m \big(P_t(i) - N_{t-1}(i)\big) =&  \sum_{i=1}^m\Big(\sum_{j=1}^t d_jx_{ij}-\sum_{j=1}^{t-1}d_jy_{ij}\Big)\\
        =& \sum_{j=1}^{t-1}d_j\Big(\sum_{i=1}^mx_{ij}- \sum_{i=1}^my_{ij}\Big)+d_t\sum_{i=1}^mx_{ij}\\
        =& \ d_t,
    \end{aligned}
    \]
    it follows that a candidate always exists as

\[
\max_{i\in[m]} \big\{P_t(i) - N_{t-1}(i)\big\}\ge d_t/m\ge \min\{d_t/m,\eps\}. \qedhere
\]
\end{proof}

The following lemma justifies the definition of candidates:

\begin{lemma}\label{lem:lowerbound}
For any $i \in [m]$ and $t \in [n]$, $\Delta_t (i) \ge -1+\eps$.
\end{lemma}

\begin{proof}
If $i$ has not been selected at any time up to $t$, then $N_t(i) = 0$ and \[\Delta_t(i) = P_t (i) \ge 0 > -1+\eps.\] Otherwise, let $j \in [t]$ be the latest time such that $s_j=i$. Then \begin{align*}\Delta_t(i) &= P_t(i)-N_t(i) \\ &= P_t(i)-N_{j-1}(i)-d_j \\ &\ge P_j(i)-N_{j-1}(i)-d_j \\&\ge \min\{d_j/m,\eps\} -d_j \\ &\ge \min\{1/m, \eps\} -1 \\ &= -1+\eps,\end{align*} where the second inequality follows from the fact that $i$ is a candidate at time $j$.
\end{proof}

Picking a candidate with the earliest deadline also ensures the upper bound:

\begin{lemma} \label{lem:upperbound}
For any $i \in [m]$ and $t \in [n]$, $\Delta_t (i) \le 1-\eps$.
\end{lemma}

\begin{proof}
Define $I^{+} := \{ i \in [m] : \Delta_t (i) > 1-\eps\}$ and assume towards contradiction that $I^{+} \neq \emptyset$. We will need to study a carefully chosen time interval $(t',t]$ and the set of rows with deadline at most $t$ at time $t'$ to derive a contradiction. Define
$$t' := \max\{j \in [t]: P_t (s_j) < N_{j-1} (s_j) + 1- \eps\},$$
which by \eqref{eq:ddl} is the first time that we pick some row with deadline later than $t$.
\begin{claim}
    The time $t'$ is well-defined. 
\end{claim}
\begin{claimproof}
    Let $i_0\in\arg\min_{i\in [m]}\Delta_t(i)$.
    Then, since $\sum_{i\in[m]}\Delta_t(i)=0$ and $I^+\neq \emptyset$, 
     \[
     \Delta_t(i_0) < -\frac{1-\eps}{|[m] \setminus I^+|} \le  -\frac{1-\eps}{m-1} \le -\eps,
     \]
     
     where the last inequality holds for $m \ge 2$. Obviously $i_0$ has been selected by time $t$, since otherwise $\Delta_t (i_0) = P_t(i_0) \ge 0$. As such let $t_0 \in [t]$ be the latest time with $s_{t_0} = i_0$, and thus
     $$P_t(i_0) < N_t (i_0) - \eps= N_{t_0}(i_0) - \eps = N_{t_0 -1}(i_0) + d_{t_0} - \eps \le N_{t_0 -1}(i_0) +1 - \eps.$$
    In particular, $t' \ge t_0$ is well-defined.
\end{claimproof}

    Define $$I:=\{i\in [m]: P_t(i)\ge N_{t'-1}(i)+1-\eps\},$$
    which by \eqref{eq:ddl} are precisely the rows with deadlines at most $t$ at time $t'$. 
    \begin{claim}\label{claim:pick-in-I}
        $s_{t'} \not \in I$ and $s_j \in I$ for every $j \in (t',t]$.
    \end{claim}
    \begin{claimproof}
        It follows from definitions of $t'$ and $I$ that $s_{t'} \not \in I$. For an arbitrary $j\in(t',t]$, it follows from the maximality of $t'$ that $$P_t(s_j) \ge N_{j-1}(s_j)+1-\eps \ge N_{t'-1} (s_j) + 1-\eps.$$
        Therefore, $s_j\in I$.
    \end{claimproof}
    \begin{claim}\label{claim:I-nonempty}
        $I^+ \subseteq I$; in particular $I \neq \emptyset$.
    \end{claim}
\begin{claimproof}
    For every $i\in I^+$, $$P_t(i) > N_{t}(i)+1-\eps \ge  N_{t'-1}(i) +1-\eps,$$ which implies $i\in I$. Thus, $I^+\subseteq I$.
\end{claimproof}
\begin{claim}\label{claim:bound-t'}
    $P_{t'}(i)<N_{t'-1}(i)+\eps$ for every $i \in I$.
\end{claim}
\begin{claimproof}

    By the definition of $t'$ and deadline, $D_{t'}(s_{t'})>t$, which means $s_{t'} \not \in I$. By the choice of $s_{t'}$ in the algorithm, all candidates $i\in C(t')$ satisfy $D_{t'}(i)>t$. In other words, $I \cap C(t') = \emptyset$. It follows from  \eqref{eq:candidate} that for every $i\in I$, \[P_{t'}(i)<N_{t'-1}(i)+\min\{d_{t'}/m,\eps\}\le N_{t'-1}(i)+\eps.\qedhere\]
\end{claimproof}

\begin{claim}\label{claim:bound-t}
    $P_t(i)\ge N_t(i)-\eps$ for every $i \in I$.
\end{claim}
\begin{claimproof}
    For an arbitrary $i\in I$, if there does not exist $j \in (t', t]$ with $s_j=i$, then 
    \[P_t(i) \ge N_{t'-1}(i)+1-\eps  =N_t(i)+1-\eps  >N_t(i)-\eps,\]
    where we used the fact that $s_{t'}\neq i$ as $s_{t'}\notin I$. Otherwise, let $j \in (t', t]$ be the latest time with $s_j=i$. Then by the maximality of $t'$, we have \[P_t(i) \ge N_{j-1}(i) + 1 - \eps \ge N_{j-1}(i) + d_j - \eps = N_j(i) -\eps = N_t(i) - \eps.\qedhere\]\end{claimproof}

     
    
Finally, $|I| \le m-1$ since by Claim \ref{claim:I-nonempty}, $I \subseteq [m] \setminus \{s_{t'}\}$. In other words, $0 \le 1 - 2\eps |I|.$ Putting everything together,

\begin{align*}
   \sum_{j \in (t', t]} d_j 
   &\le 1 - 2\eps |I| + \sum_{j \in (t', t]} d_j\\ 
   &= 1 - 2\eps|I| + \sum_{i\in I} (N_t(i) -N_{t'-1}(i)) \tag{\text{Claim} \ref{claim:pick-in-I}} \\
   &=1 + \sum_{i\in I} (N_t(i)-\eps -N_{t'-1}(i)-\eps)\\ 
   &< \sum_{i\in I} (P_t(i)-N_{t'-1}(i)-\eps) \tag{\text{Claim} \ref{claim:bound-t}+\text{Claim} \ref{claim:I-nonempty}}\\ 
    &< \sum_{i\in I} (P_t(i)-P_{t'}(i)) \tag{\text{Claim} \ref{claim:bound-t'}+\text{Claim} \ref{claim:I-nonempty}}\\
    &\le \sum_{i \in [m]} (P_t(i)-P_{t'}(i))\\
    &=\sum_{j \in (t', t]} d_j,
\end{align*}
This is a contradiction, from which we conclude $I^+ = \emptyset$. \qedhere
\end{proof}

\begin{proof}[Proof of Theorem \ref{thm:WCAP}]
The bounds follow from Lemmas~\ref{lem:lowerbound} and~\ref{lem:upperbound}. It remains to show that the algorithm can be implemented in linear time. We may precompute $P_t (i)$ for all $t \in [n], i \in [m]$ via prefix sum, and we may update $N_t(i)$ for all $i \in [m]$ after the $t$-th iteration. The deadlines $D_t (i)$ may also be computed in linear amortized time by incrementing them each iteration until $P_T(i) \ge N_{t-1}(i) + 1-\eps$ as $T \le n$. 
\end{proof}

\section{Application for flow-time scheduling}\label{sec:schedule}

Let $M= [m]$ be a set of machines, and $J=[n]$ be a set of jobs. Each job $j\in J$ can be scheduled on a subset of machines $M_j\subseteq M$. Each job $j\in [n]$ has a processing time $d_j$ and release time $r_j$. Let $C_j$ be the completion time of $j$. Order the jobs by their release times $r_1\le r_2\le\cdots \le r_n$. We want to schedule each job to a machine such that the maximum flow-time $\max_{j\in[n]} (C_j-r_j)$ is minimized. Consider the following linear programming relaxation for this problem:
\begin{equation}\label{eq:max_flow-time}
\begin{aligned}
\min~&\ T\\
s.t.~ &\ \sum_{i=1}^m x_{ij}=1\ \forall j\in [n]\\
&\sum_{j=s}^t x_{ij}d_{j}\le (r_t-r_s)+T,\ \forall i\in [m],\ 1\le s\le t\le n\\
&x\ge 0\\
&x_{ij}=0,\ \forall i,j: i\notin M_j.
\end{aligned}
\end{equation}

Morell and Skutella \cite{MorellSkutella2020IPCO} show that Conjecture \ref{conj:weighted_carpool} would imply the integrality gap of \eqref{eq:max_flow-time} is at most $3$. Efficiently computing such an assignment $y$ would lead to a $3$-approximation algorithm for maximum flow-time scheduling. It follows readily from their proof that upper-bounding $\sum_{j=s}^t d_j(y_{ij}-x_{ij})$ for every machine $i$ and interval of jobs defined by $1\le s\le t\le n$ suffices to prove the desired approximation ratio. We summarize this in the following lemma and give its proof for completeness.
\begin{lemma}[\cite{MorellSkutella2020IPCO}]\label{lemma:flow-time}
    Let $m, n$ be positive integers and $d \in \R^n_{>0}$. If for every fractional assignment $x \in [0,1]^{m\times n}$, there is an assignment $y\in \{0,1\}^{m\times n}$ so that $y_{ij} = 0$ whenever $x_{ij} = 0$ and
    \begin{equation}\label{eq:interval_disc}
        \sum_{j=s}^t d_j(y_{ij}-x_{ij})\le \alpha\cdot \max_{j\in [n]} d_j\quad \forall i\in[m], 1\le s\le t\le n,
    \end{equation}
    then the integrality gap of the linear programming relaxation \eqref{eq:max_flow-time} is at most $(\alpha+1)$. A polynomial time algorithm to compute such an assignment $y$ would lead to an $(\alpha+1)$-approximation algorithm for maximum flow-time scheduling.
\end{lemma}
\begin{proof}
Denote $d_{\max}:=\max_{j\in[n]}d_j$. Let $(x,T)$ be the optimal solution to the linear program \eqref{eq:max_flow-time}. Let $\OPT$ be the minimum maximum flow-time. Suppose $y$ is the assignment satisfying \eqref{eq:interval_disc}. Let $i\in [m]$ and $1\le s\le t\le n$. The total processing time of jobs assigned to machine $i$ and released between $r_s$ and $r_t$ is
\[
\begin{aligned}
\sum_{j=s}^t y_{ij}d_j\le& \sum_{j=s}^t x_{ij}d_j+\alpha\cdot d_{\max}\\
\le& (r_t-r_s)+T+\alpha\cdot \OPT\\
\le& (r_t-r_s)+(1+\alpha)\OPT,
\end{aligned}
\]
where the first inequality follows from \eqref{eq:interval_disc}. The second inequality follows from the fact that the job of size $d_{\max}$ has flow-time at least $d_{\max}$. The last inequality follows from the fact that $T\le \OPT$. Suppose job $t$ is scheduled on $i$. Let $r_s$ be the latest time before $t$ such that $i$ is idle when some job $s$ scheduled to $i$ is released (which exists because $i$ is idle when the first job scheduled to $i$ is released). Then, the completion time of $t$ is $C_t=r_s+\sum_{j=s}^t y_{ij}d_j$. The flow-time of $t$ is $C_t-r_t=r_s+\sum_{j=s}^t y_{ij}d_j-r_t\le (1+\alpha)\OPT$. Therefore, the maximum flow-time is at most $(1+\alpha)\OPT$.
\end{proof}

Previously, a $(3-\frac{2}{m})$-approximation algorithm was known for identical machines \cite{BenderChakrabartiMuthukrishnan1998SODA, Mastrolilli2004IJFCS}. Combining Theorem \ref{thm:WCAP} and Lemma \ref{lemma:flow-time}, we get a $(3-\frac{1}{m-1})$-approximation algorithm for this setting, matching the current best approximation ratio asymptotically. Yet, their algorithm (FIFO) follows a much simpler greedy approach, which orders jobs by their release times and assigns them online to the machine with the least remaining processing time of jobs that have been assigned to it.

\begin{figure}[htbp]
    \centering
\begin{tikzpicture}[x=0.35cm,y=1cm]
\def\rowh{1.0}

\def\Wleft{17.6}     
\def\Wright{9.6}   
\def\dx{22}        

\def\wg{2.4}       
\def\wo{3.2}       
\def\blueW{9.6}    
\def\wdots{2.4}    

\definecolor{gree}{RGB}{92,166,106}
\definecolor{oran}{RGB}{219,144,63}
\definecolor{blu}{RGB}{35,152,207}

\tikzset{
  ttl/.style={fill=orange!85!black,rounded corners=2pt,inner xsep=8pt,
              inner ysep=3pt,font=\bfseries\large,text=black},
  frac/.style={font=\normalsize},
  labc/.style={anchor=east,font=\large},
  rlab/.style={font=\normalsize},
  sep/.style={line width=0.4pt}
}

\node[rlab] at ({1},{4*\rowh+0.55}) {$r=\delta$};
\node[rlab] at ({4.2},{4*\rowh+0.55}) {$r=2\delta$};
\node[rlab]            at ({\wg+\wo+1.3},{4*\rowh+0.5}) {$\cdots$};
\node[rlab]  at ({\wg+\wo+\wdots+\blueW/2},{4*\rowh+0.55}) {$r=m\delta$};

\node[labc] at ({-0.8},{3.5*\rowh}) {$b_1=\delta$};
\node[labc] at ({-0.8},{2.5*\rowh}) {$b_2=2\delta$};
\node[labc] at ({-0.8},{1.5*\rowh}) {$\vdots$};
\node[labc] at ({-0.8},{0.5*\rowh}) {$b_m=m\delta$};

\node[frac] at ({\wg/2},{3.5*\rowh}) {$\dfrac{1}{m}$};
\draw[sep] ({\wg},{3*\rowh}) -- ({\wg},{4*\rowh});

\node[frac] at ({\wg/2},{2.5*\rowh}) {$\dfrac{1}{m}$};
\node[frac] at ({\wg+\wo/2},{2.5*\rowh}) {$\dfrac{1}{m-1}$};
\draw[sep] ({\wg},{2*\rowh}) -- ({\wg},{3*\rowh});
\draw[sep] ({\wg+\wo},{2*\rowh}) -- ({\wg+\wo},{3*\rowh});

\node[frac] at ({\wg/2},{1.5*\rowh}) {$\dfrac{1}{m}$};
\node[frac] at ({\wg+\wo/2},{1.5*\rowh}) {$\dfrac{1}{m-1}$};
\node[frac] at ({\wg+\wo+1.3},{1.5*\rowh}) {$\cdots$};
\draw[sep] ({\wg},{1*\rowh}) -- ({\wg},{2*\rowh});

\draw[sep] ({\wg+\wo},{1*\rowh}) -- ({\wg+\wo},{2*\rowh});

\node[frac] at ({\wg/2},{0.5*\rowh}) {$\dfrac{1}{m}$};
\node[frac] at ({\wg+\wo/2},{0.5*\rowh}) {$\dfrac{1}{m-1}$};
\node[frac] at ({\wg+\wo+1.3},{0.5*\rowh}) {$\cdots$};
\node[font=\Large] at ({\wg+\wo+\wdots+\blueW/2},{0.5*\rowh}) {$1$};
\draw[sep] ({\wg},{0}) -- ({\wg},{1*\rowh});
\draw[sep] ({\wg+\wo+\wdots},{0}) -- ({\wg+\wo+\wdots},{1*\rowh});

\draw[sep] ({\wg+\wo},{0}) -- ({\wg+\wo},{1*\rowh});

\draw[thick] ({0},{0}) rectangle ({\Wleft},{4*\rowh});
\foreach \r in {1,2,3} \draw[thick] ({0},{\r*\rowh}) -- ({\Wleft},{\r*\rowh});


\foreach \k in {0,1,2,3}{
  \node[frac] at ({\dx+(\k+0.5)*\wg},{3.5*\rowh}) {$\dfrac{1}{m}$};
}
\draw[sep] ({\dx+\wg},{3*\rowh}) -- ({\dx+\wg},{4*\rowh});
\draw[sep] ({\dx+2*\wg},{3*\rowh}) -- ({\dx+2*\wg},{4*\rowh});
\draw[sep] ({\dx+3*\wg},{3*\rowh}) -- ({\dx+3*\wg},{4*\rowh});

\foreach \k in {0,1,2}{
  \node[frac] at ({\dx+(\k+0.5)*\wo},{2.5*\rowh}) {$\dfrac{1}{m-1}$};
}
\draw[sep] ({\dx+\wo},{2*\rowh}) -- ({\dx+\wo},{3*\rowh});
\draw[sep] ({\dx+2*\wo},{2*\rowh}) -- ({\dx+2*\wo},{3*\rowh});

\node[frac] at ({\dx+\Wright/2},{1.5*\rowh}) {$\cdots$};

\node[font=\Large] at ({\dx+\Wright/2},{0.5*\rowh}) {$1$};

\draw[thick] ({\dx},{0}) rectangle ({\dx+\Wright},{4*\rowh});
\foreach \r in {1,2,3} \draw[thick] ({\dx},{\r*\rowh}) -- ({\dx+\Wright},{\r*\rowh});
\end{tikzpicture}
\caption{There are $m$ machines $M=[m]$ with closing times $b_i=i\cdot \delta$ for some $\delta \ll \tfrac{1}{m}$. There are $m$ batches of jobs, released at times $\delta, 2\delta,...,m\delta$. The $j$-th batch has $(m-j+1)$ jobs, each with processing time $\frac{1}{m-j+1}$. Left: FIFO schedules the $j$-th batch to machines $j,j+1,...,m$, one for each, because those are the machines that are not closed yet. The maximum flow-time is $\Big(\sum_{j\in [m]} \tfrac{1}{j}\Big) - m\delta=\Omega(\log m)$. Right: OPT schedules the $j$-th batch to machine $j$, which is feasible because the $j$-th batch is released no later than machine $j$ is closed. The maximum flow-time is $1-m\delta \le 1$.}
\label{fig:FIFO}
\end{figure}
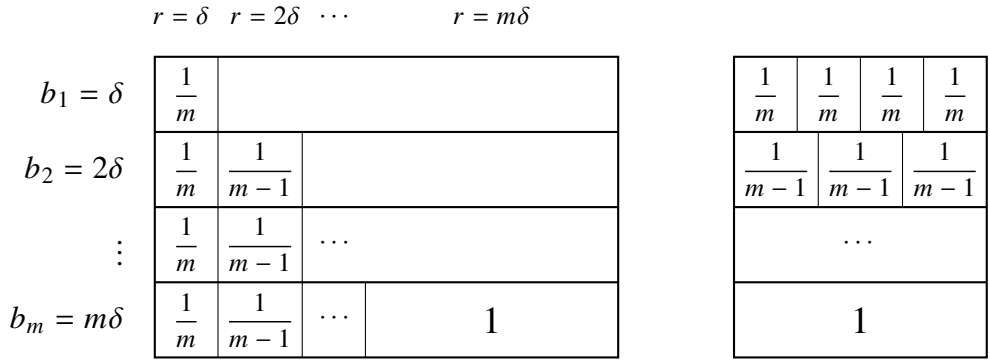

We study a variation where each machine $i$ has a closing time $b_i$, and a job can be scheduled on machine $i$ only if $r_j\le b_i$. We begin by observing that FIFO has a lower bound $\Omega(\log m)$ for the setting of closing times (see Figure \ref{fig:FIFO}).
We give a $(3-\frac{1}{m-1})$-approximation algorithm for the setting of closing times.  We start by observing that Algorithm \ref{alg:deadline} has the following property.
\begin{prop}\label{prop:open_time}
    Given $d \in \R^n_{>0}$,  $a\in [n]^m$ and a fractional assignment $x\in [0,1]^{m\times n}$ so that $x_{ij}=0\ \forall i,j: j<a_i$, there is a linear-time algorithm that computes an assignment $y\in\{0,1\}^{m\times n}$ so that $y_{ij}=0\ \forall i,j:j<a_i$ and
    \[ \Big|\sum_{j=1}^t d_j (x_{ij} - y_{ij}) \Big| \le \Big(1-\frac{1}{2m-2}\Big) \cdot \max_{j \in [n]} d_j\quad \forall i\in[m], t \in [n].\]
\end{prop}
\begin{proof}
    Note that for every $t<a_i$, $P_t(i)=\sum_{j=1}^td_jx_{ij}=0$, since $x_{ij}=0\ \forall j<a_i$. Therefore, for every $t<a_i$, $P_t(i)=0<N_{t-1}(i)+\min\{d_t/m,\eps\}$, which means $i$ is not a candidate at time $t$. Thus, the assignment $y$ returned by Algorithm \ref{alg:deadline} satisfies $y_{ij}=0\ \forall j<a_i$.
\end{proof}

\begin{lemma}\label{lemma:closing_time}
    Given release times $r\in \R_{\ge 0}^n$ of jobs, closing times $b\in \R_{\ge 0}^m$ of machines, $d \in \R^n_{>0}$ and a fractional assignment $x\in [0,1]^{m\times n}$ so that $x_{ij}=0\ \forall i,j: r_j>b_i$, there is a linear-time algorithm that computes an assignment $y\in\{0,1\}^{m\times n}$ so that $y_{ij}=0\ \forall i,j: r_j>b_i$ and 
    \[ \sum_{j=s}^t d_j(y_{ij}-x_{ij})\le  \Big(2-\frac{1}{m-1}\Big) \cdot \max_{j \in [n]} d_j \quad \forall i\in[m], 1\le s\le t\le n.\]
\end{lemma}
\begin{proof}
    For every $i\in [m]$, define $a_i:=\max\{j:r_j\le b_i\}$. Then, assignment $x$ satisfies $x_{ij}=0\ \forall j>a_i$. Applying Proposition \ref{prop:open_time} with $x'_{ij}:=x_{i,n-j+1}$ and $a'_i:=n-a_i+1$, we conclude that there exists an assignment $y'$ satisfying $y'_{ij}=0\ \forall j<a'_i$, and thus $y_{ij}:=y'_{i,n-j+1}$ satisfies $y_{ij}=0\ \forall j>a_i$. Moreover, for every  $i\in[m]$ and $ 1\le s\le t\le n$,
    \[
    \begin{aligned}
    \sum_{j=s}^t d_j(y_{ij}-x_{ij})=&\sum_{j=n-t+1}^{n-s+1} d_j(y'_{ij}-x'_{ij})\\
    =& \ \sum_{j=1}^{n-s+1} d_j(y'_{ij}-x'_{ij})-\sum_{j=1}^{n-t} d_j(y'_{ij}-x'_{ij})\\
    \le& \ 2\Big(1-\frac{1}{2m-2}\Big)\cdot \max_{j\in [n]} d_j\\
    =& \ \Big(2-\frac{1}{m-1}\Big)\cdot \max_{j\in [n]} d_j.
    \end{aligned}
    \]
\end{proof}

\begin{proof}[Proof of Theorem \ref{thm:closing_time}]
    For every $j\in J$, let $M_j=\{i\in[m]: r_j\le b_i\}$. It follows from Lemmas \ref{lemma:closing_time} and \ref{lemma:flow-time} that there is a $(3-\frac{1}{m-1})$-approximation algorithm for maximum flow-time scheduling with closing times.
\end{proof}

\section{Lower bounds}\label{sec:lower_bound}

\begin{proof}[Proof of Proposition \ref{prop:caplb}]

If $m = 2$ we can simply take $x_{1,1} = x_{2,1} = \tfrac{1}{2}$. Otherwise, let $x \in \R^{m \times (m-1)}$ be defined as $x_{i,1} = \frac{1}{2m-2}$ for $i < m$ and $x_{m,1} = \tfrac{1}{2}$, $x_{m,j} = 0$ for $j > 1$ and $x_{i,j} = \frac{1}{m-1}$ otherwise (see below). Assume towards contradiction that there exists some assignment $y \in \{0,1\}^{m \times (m-1)}$ so that  $|\sum_{j \in [t]} (x_{ij} - y_{ij})| < 1-\frac{1}{2m-2}\ \forall i\in[m], t \in [m-1]$. If $y_{i,1} = 1$ for some $i \neq m$ then already $|x_{i,1} - y_{i,1}| = 1-\frac{1}{2m-2}$. Otherwise, $y_{m,1} = 1$ and yet each $i\in[m-1]$ must be assigned at least once, since $\sum_{j\in [m-1]} x_{ij}=1 - \frac{1}{2m-2}$. But there are only $m-2$ columns left, so one of $i\in[m-1]$ must never be assigned, which is a contradiction. \qedhere

\begin{center}
\begin{tikzpicture}[baseline]
\matrix (M) [matrix of math nodes, nodes in empty cells,
             row sep=1pt, column sep=0pt]{
            & 1 & 2 & \cdots & m\!-\!1 \\[1pt]
1           & \dfrac{1}{2m-2} & \dfrac{1}{m-1} & \cdots & \dfrac{1}{m-1} \\
2           & \dfrac{1}{2m-2} & \dfrac{1}{m-1} & \cdots & {\boxed{\dfrac{1}{m-1}}} \\
\vdots      & \vdots & \vdots & \ddots & \vdots \\
m\!-\!1     & \dfrac{1}{2m-2} &
              {\boxed{\dfrac{1}{m-1}}} & \cdots & \dfrac{1}{m-1} \\
m           & {\boxed{\dfrac{1}{2}}} & 0 & \cdots & 0 \\
            & & & & \\ 
};
\draw[line width=0.5pt]
  (M-2-2.north west) rectangle ([xshift=1.5em,yshift=0.9em]M-7-5.south east);
\end{tikzpicture}
\end{center}
\end{proof}

\begin{proof}[Proof of Proposition \ref{prop:carlb}] 

Consider the fractional assignment $x\in[0,1]^{3\times n}$ below where there are $n= 1+2 \cdot p$ columns, with $p := \lceil \tfrac{1-\delta}{2\delta} \rceil$ repeating pairs of columns after the first. Assume towards contradiction that there exists some assignment $y \in \{0,1\}^{3 \times n}$  so that $|\sum_{j \in [t]} (x_{ij} - y_{ij})| < 1-\delta$ for all $t \in [n]$. We have to assign $y_{3,1}=1$, for otherwise already $|x_{3,1}-y_{3,1}| = 1-\delta$. Since $\sum_{j\in [2s]} x_{1,j}=s-\delta$, to keep the discrepancy of the first row bounded, we must assign $y_{1, 2s} = 1$ and $y_{1,2s+1} =0$ for $s \in [p]$. 
Then we are forced to assign $y_{3,2s+1}=1$ for every $s\in [p]$ because $x_{2,2s+1}=0$ implies $y_{2,2s+1} =0$. Yet the second row sums to $2\delta \cdot p \ge 1-\delta$ without ever being assigned, a contradiction.\qedhere 

\begin{center}
\begin{tikzpicture}
\matrix (M) [matrix of math nodes, nodes in empty cells,
             row sep=1pt, column sep=5pt,
             left delimiter={[}, right delimiter={]}]{
  \delta     & 1-2\delta & 2\delta    & \cdots & 1-2\delta & 2\delta \\
  0          & 2\delta   & 0          & \cdots & 2\delta   & 0       \\
  1-\delta   & 0         & 1-2\delta  & \cdots & 0         & 1-2\delta \\
};
\end{tikzpicture}
\end{center}
\end{proof}

\begin{proof}[Proof of Proposition \ref{prop:intlb}] A construction is given by $x_{i,j} = v_i$ for all $i \in [3], j \in [100]$ where $v = (0.01, 0.48, 0.51)$, which can be computationally verified to achieve interval discrepancy equal to 1.32. See Appendix \ref{sec:code}. 
\qedhere
\end{proof}

\section{Open problems}\label{sec:open}

We ask for a version with non-uniform weights which may shed light on the \emph{unrelated machines} scheduling, where job $j$ has processing time $d_{ij}$ on machine $i$:

\begin{conjecture}\label{conj:non_uniform}
   Let $m, n$ be positive integers and $d \in \R^{m \times n}_{>0}$. For every fractional assignment $x \in [0,1]^{m \times n}$, there is an assignment $y \in \{0,1\}^{m \times n}$ so that for every $i\in[m]$ and $t\in[n]$,
    \[
    \Big|\sum_{j \in [t]} d_{ij} (x_{ij}-y_{ij}) \Big| \le  \max_{i\in [m], j \in [n]} d_{ij}.
    \]
\end{conjecture}

 The construction in Proposition~\ref{prop:carlb} applied to $d_{ij} \in \{0,d_j\}$ also shows a lower bound of $(1-\delta)\cdot \max_{i\in [m], j \in [n]} d_{ij}$ for all $\delta > 0$ for this setting. The argument in~\cite{lenstra1990approximation} shows that the bound holds for the non-prefix version, i.e. $t = n$.

We also formulate a \emph{weighted committee assignment problem}:

\begin{conjecture}\label{conj:int_weights}
   Let $m, n$ be positive integers and $d \in \R^{n}_{>0}$. For every $x \in [0,1]^{m \times n}$ so that $n_j := \sum_{i \in [m]} x_{ij} \in \mathbb{N}$ for all $j \in [n]$, there is $y \in \{0,1\}^{m \times n}$ so that $\sum_{i \in [m]} y_{ij} = n_j$ for all $j \in [n]$ and
    \[
    \Big|\sum_{j \in [t]} d_{j} (x_{ij}-y_{ij}) \Big| \le  \max_{j \in [n]} d_{j}\quad \forall i\in[m], t\in [n].
    \]
\end{conjecture}

It is plausible that Conjectures~\ref{conj:weighted_carpool},~\ref{conj:non_uniform} and~\ref{conj:int_weights} may all be combined:

\begin{conjecture}\label{conj:allinone}
   Let $m, n$ be positive integers and $d \in \R^{m \times n}_{>0}$. For every $x \in [0,1]^{m \times n}$ so that $n_j := \sum_{i \in [m]} x_{ij} \in \mathbb{N}$ for all $j \in [n]$, there is $y \in \{0,1\}^{m \times n}$ so that $\sum_{i \in [m]} y_{ij} = n_j$ for all $j \in [n]$, $y_{ij} = 0$ whenever $x_{ij} = 0$, and
    \[
    \Big|\sum_{j \in [t]} d_{ij} (x_{ij}-y_{ij}) \Big| \le  \max_{i \in [m], j \in [n]} d_{ij}\quad \forall i\in[m], t\in [n].
    \]
\end{conjecture}

\paragraph{Acknowledgements} We would like to thank R. Ravi for numerous discussions. The work was done during an internship of the first author at Microsoft Research.

\bibliographystyle{alpha} 
\bibliography{cap}

\appendix


\section{Computer program}\label{sec:code}

We use the following Julia function to verify Proposition~\ref{prop:intlb}:

\begin{center}
\begin{minipage}{\linewidth}
\lstset{language=Julia, numbers=left, basicstyle=\ttfamily\small}
\begin{lstlisting}
using JuMP, Gurobi

function interval_disc(m,n,x)
    opt = Model(Gurobi.Optimizer)
    @variable(opt, Delta >= 0) # interval discrepancy
    @variable(opt, y[1:m,1:n], Bin) # assignment 
    
    for j in 1:n
        @constraint(opt, sum(y[:,j]) == 1)
    end
    for i in 1:m
        for s in 1:n
            for t in s:n
                @constraint(opt, sum(y[i,s:t]) - sum(x[i,s:t]) <= Delta)
                @constraint(opt, sum(y[i,s:t]) - sum(x[i,s:t]) >= -Delta)
            end
        end
    end
    
    @objective(opt, Min, Delta)
    optimize!(opt)
    return value(Delta)
end
\end{lstlisting}
{  \captionsetup{hypcap=false}
\captionof{lstlisting}{Input: $m,n \in \mathbb{N}$ and a fractional assignment $x \in [0,1]^{m\times n}$. Output: the smallest $\Delta$ for which there exists an assignment $y \in \{0,1\}^{m\times n}$ so that $|\sum_{j \in [s,t]} (x_{ij}-y_{ij})| \le \Delta$ for all $i \in [m]$ and $1 \le s \le t \le n$.}
}
\label{lst:intervaldisc}
\end{minipage}
\end{center}
\end{document}